\documentclass[12pt]{amsart}
\usepackage[margin = 1in]{geometry}
\usepackage{subfig, graphicx}
\begin{document}
\title[Zero-Coupon Treasury Rates with VIX]{Zero-Coupon Treasury Rates and Returns using the Volatility Index}
\author{Jihyun Park, Andrey Sarantsev}
\address{University of Michigan, Ann Arbor Medical Center}
\email{jihyunp@med.umich.edu}
\address{University of Nevada, Reno, Department of Mathematics \& Statistics}
\email{asarantsev@unr.edu}
\date{\today}

\newtheorem{theorem}{Theorem}
\newtheorem{lemma}{Lemma}
\theoremstyle{definition}
\newtheorem{asmp}{Assumption}
\newtheorem{definition}{Definition}

\begin{abstract}
We study a multivariate autoregressive stochastic volatility model for the first 3 principal components (level, slope, curvature) of 10 series of zero-coupon Treasury bond rates with maturities from 1 to 10 years. We fit this model using monthly data from 1990. Unlike classic models with hidden stochastic volatility, here it is observed as VIX: the volatility index for the S\&P 500 stock market index. Surprisingly, this stock index volatility works for Treasury bonds, too. Next, we prove long-term stability and the Law of Large Numbers. We express total returns of zero-coupon bonds using these principal components. We prove the Law of Large Numbers for these returns. All results are done for discrete and continuous time. 
\end{abstract}
\maketitle
\thispagestyle{empty}

\section{Introduction}

Take 10 series of zero-coupon Treasury bond rates with maturities 1, 2, \ldots, 10 years: Monthly end-of-month data. Apply the principal component (PC) analysis to these series to reduce dimension and take the first 3 PC: level, slope, curvature. They explain almost all variance; in fact, level alone explains more than 95\%. Such dimension reduction is the standard procedure in Treasury bond market analysis. Next, take another series VIX, the Volatility Index of Standard \& Poor 500 (an index of large US stocks), measuring daily implied volatility. Denote the average monthly VIX data by $V(t)$. 

We aim to construct a stochastic volatility model for the vector $\mathbf{X}$ of PC: Instead of implied and unobserved volatility as in GARCH or classic stochastic volatility models, we use the observed volatility $V(t)$ for the American stock market, which makes estimation much easier. We analyze the long-term properties of this model in general dimension $d$, although the statistical analysis is carried only for 1, 2, or 3 dimensions. It is remarkable that volatility implied from the stock market works well for bond rates as well. 

First, we describe the model in discrete time. Consider a sequence of independent identically distributed (IID) $\mathbb R^{d+1}$-valued random vectors with mean zero: $(Z_0(t), Z_1(t), \ldots, Z_d(t))$. Denote $\mathbf{Z}(t) = (Z_1(t), \ldots, Z_d(t))$. Then let 
\begin{equation}
\label{eq:AR-SV}
\mathbf{X}(t) = \mathbf{a} + \mathbf{B}\mathbf{X}(t-1) + \mathbf{c}V(t) + \xi(t)\mathbf{Z}(t),
\end{equation}
where $\mathbf{a}, \mathbf{c} \in \mathbb R^d$ are constant vectors, $\mathbf{B}$ is a constant $d\times d$-matrix, and $V(t)$ is modeled as an autoregression of order $1$ on the logarithmic scale:
\begin{equation}
\label{eq:log-AR}
\ln V(t) = \alpha + \beta\ln V(t-1) + Z_0(t),
\end{equation}
the $d\times d$-matrix $\xi$ is diagonal, with $p$ elements $V(t)$ and $d-p$ unit elements: 
\begin{equation}
\label{eq:diag}
\xi(t) = \mathrm{diag}(V(t), \ldots, V(t), 1, \ldots, 1).
\end{equation}
In continuous time, the model~\eqref{eq:AR-SV} and~\eqref{eq:log-AR} can be written as a stochastic differential equation (SDE) with respect to a $d$-dimensional process $\mathbf{X} = (\mathbf{X}(t),\, t \ge 0)$:
\begin{equation}
\label{eq:OU-SV}
\mathrm{d}\mathbf{X}(t) = (\mathbf{a} - \mathbf{B}\mathbf{X}(t) + \mathbf{c}V(t))\,\mathrm{d}t + \xi(t)\,\mathrm{d}\mathbf{W}(t),
\end{equation}
where $V = (V(t),\, t \ge 0)$ is modeled as an Ornstein-Uhlenbeck process on the log scale:
\begin{equation}
\label{eq:log-OU}
\mathrm{d}\ln V(t) = (\alpha - \beta \ln V(t))\,\mathrm{d}t + \mathrm{d}W_0(t),
\end{equation}
the matrix $\xi(t)$ is defined as in~\eqref{eq:diag}, and $(W_0, \mathbf{W}) = (W_0, W_1, \ldots, W_d)$ is a $(d+1)$-dimensional Brownian motion with drift vector zero and covariance matrix $\Sigma$. 

If we had $\mathbf{c} = 0$ and $\xi(t) = I_d$, the $d\times d$-identity matrix, the model~\eqref{eq:AR-SV},~\eqref{eq:log-AR},~\eqref{eq:diag} would be a classic linear vector autoregression of order 1, and the model~\eqref{eq:diag},~\eqref{eq:OU-SV} and~\eqref{eq:log-OU} would be a multivariate Ornstein-Uhlenbeck process. Multiplying the noise terms by $V$ makes volatility non-constant and stochastic. Therefore, we call the discrete-time model~\eqref{eq:AR-SV},~\eqref{eq:log-AR},~\eqref{eq:diag} a {\it multivariate autoregressive stochastic volatility model}, and the continuous-time model~\eqref{eq:diag},~\eqref{eq:OU-SV},~\eqref{eq:log-OU} a {\it multivariate Ornstein-Uhlenbeck process with stochastic volatility}. 

In Section 2, we describe the data and perform statistical analysis. We also provide background and motivation, including historical review and connections with existing research. Why choose zero-coupon Treasury bonds instead of the classic Treasury bonds? Classic bonds pay semiannual coupons; and zero-coupon bonds pay only principal at maturity, which makes it easier to compute total returns. The stochastic volatility is directly observed as VIX data, in contrast with classic stochastic volatility models with hidden volatility. We fit the autoregression model~\eqref{eq:log-AR} for $V(t)$. Then we fit a univariate or multivariate version of the model~\eqref{eq:AR-SV} for the PC. We verify whether including this volatility term improves fit. 

In Section 3, we state and prove long-term stability results: Theorems~\ref{thm:stationary},~\ref{thm:ergodicity},~\ref{thm:LLN}, for the discrete-time model~\eqref{eq:AR-SV},~\eqref{eq:log-AR},~\eqref{eq:diag}, under assumptions that the spectral radius of the matrix $\mathbb{B}$ is less than 1, and $\beta \in (0, 1)$, together with a few other technical assumptions. We state and prove a similar result in Theorem~\ref{thm:cont-time} for the continuous-time diffusion model~\eqref{eq:diag},~\eqref{eq:OU-SV},~\eqref{eq:log-OU}. 

In Section 4, we approximate total bond returns as a simple linear function of the vector of PC. We state and prove the Law of Large Numbers in Theorems~\ref{thm:returns-discrete} and~\ref{thm:returns-continuous}. Finally, we discuss the dependence of term premium (difference between total bond returns and risk-free returns) upon maturity; and test the applicability of the Capital Asset Pricing Model (CAPM) adjusted for the bond market. The first PC (level) is similar to the market exposure factor for stocks in CAPM. Other PC (slope and curvature) are additional factors. 

All data are taken from Federal Reserve Economic Data (FRED) web site. We published this data in a single Excel file together with Python code on \texttt{GitHub} repository \texttt{asarantsev/vix-zeros} with the code generating all data and plots in this article, and some other plots not included in this text.  

\section{Background, Motivation, and Financial Data}

\subsection{Notation and definitions} For $\mathbf{x} = (x_1, \ldots, x_d) \in \mathbb R^d$, the {\it $L^2$-norm} is defined as 
$$
|\mathbf{x}|_2 := \left[|x_1|^2 + \ldots + |x_d|^2\right]^{1/2}.
$$
Next, following \cite{Matrix}, we define two matrix norms for $d\times d$ real-valued matrices. Denote by $\sigma(A)$ the set of eigenvalues of the matrix $A$. The {\it spectral norm} induced by the $L^2$-norm in $\mathbb R^d$, see \cite[Example 5.6.6]{Matrix}:
$$
|\!|A|\!|_S := \sup\limits_{\mathbf{x} \in \mathbb R^d\setminus\{\mathbf{0}\}}\frac{|A\mathbf{x}|_2}{|\mathbf{x}|_2} = \left[\max(\sigma(A^TA))\right]^{1/2}.
$$
The {\it Frobenius norm} \cite[(5.2.7), (5.6.0.2)]{Matrix} is $|\!|A|\!|_F^2 := \mathrm{tr}(A^TA) = \sum_{\lambda \in \sigma(A^TA)}\lambda$. From the definition of these two matrix norms, we compare them: 
\begin{equation}
\label{eq:norms}
|\!|A|\!|_S \le |\!|A|\!|_F \le \sqrt{d}|\!|A|\!|_S.
\end{equation}
The Frobenius norm corresponds to the inner product in the space of $d\times d$-matrices: $(A, B)_F := \mathrm{tr}(A^TB)$. Therefore, we have the Cauchy inequality:
\begin{equation}
\label{eq:Cauchy}
|(A, B)_F| \le |\!|A|\!|_F\cdot |\!|B|\!|_F.
\end{equation}
For symmetric nonnegative definite $A$, we have (with sum over $\lambda \in \sigma(A)$):
\begin{equation}
\label{eq:positive-F}
|\!|A|\!|_F^2 = \sum \lambda^2 \le d\cdot(\sum \lambda)^2 = d\cdot\mathrm{tr}^2(A).
\end{equation}

The following definitions are taken from the monograph \cite{Stability} for discrete time and from the article \cite{MT1993a} for continuous time. We consider a discrete- or continuous-time Markov process $\mathbf{X} = (\mathbf{X}(t), t \ge 0)$ in $\mathbb R^d$, with time $t = 0, 1, 2, \ldots$ for discrete time and $t \in [0, \infty)$ for continuous time. We assume the process is {\it time homogeneous:} $\mathbb P(\mathbf{X}(t) \in A\mid \mathbf{X}(s) = x)$ depends only on $t - s$, not on $t, s$. 

This Markov process has a {\it stationary distribution} ({\it invariant probability measure}) $\pi$ on $\mathbb R^d$ if $X(0) \sim \pi$ implies $X(t) \sim \pi$ for all $t \ge 0$. This process is called {\it ergodic} if it has a unique stationary distribution, and regardless of the initial condition $\mathbf{x} \in \mathbb R^d$, converges to this stationary distribution as time tends to infinity:
$$
\sup\limits_{A \subseteq \mathbb R^d}|\mathbb P(\mathbf{X}(t) \in A\mid \mathbf{X}(0) = x) - \pi(A)| \to 0,\quad t \to \infty.
$$
We say that this Markov process is {\it tight} if, regardless of the initial condition $\mathbf{X}(0)$, for every $\eta > 0$, there exists a compact set $\mathcal K \subseteq \mathbb R^d$ such that for every $t \ge 0$, $\mathbb P(\mathbf{X}(t) \in \mathcal K) \ge 1 - \eta$.

\subsection{Data description} Taken from Federal Reserve Economic Data (FRED) web site, monthly January 1990--August 2024. 

\begin{itemize}
\item Zero-coupon Treasury rates with maturities 1, 2, \ldots, 10 years: end-of-month, FRED series THREEFY1, THREEF2, \ldots, THREEFY10. 
\item Chicago Board of Options Exchange (CBOE) Volatility Index (VIX) for Standard \& Poor (S\&P) 500 stock market index, monthly aveage: FRED series VIXCLS.
\end{itemize}

\subsection{VIX: the volatility index} The two main investment modes are ownership and debt (loans). In the USA, the most liquid publicly owned forms of each are large public stocks (such that comprise the celebrated Standard \& Poor 500 index) and Treasury bonds with various maturities. In this article, we reveal an unexpected connection between these two. It comes as no surprise that these must be dependent upon each other, because they are competing for investors' capital. In this article, we study one aspect of this relationship: Treasury bond market with volatility measured using stock market volatility.

In a simple approximation, the index S\&P 500 moves daily/weekly/monthly (on the log scale) as a random walk, meaning percentage changes are independent and identically distributed. However, in practice, this is not quite true. Instead, there are calm and turbulent periods with lower and higher standard deviation (called {\it volatility} in quantitative finance) of index fluctuations. Denote volatility for month $t$ by $V(t)$. If $X(t)$ are index returns, then we can express $X(t) = V(t)Z(t)$, where $Z(t)$ are IID with mean zero and variance $1$. We measure $V(t)$ using the volatility index (VIX) maintained by the Chicago Board of Options Exchange (CBOE). This index is made for the Standard \& Poor 500: a well-known index of large 500 US companies, used as American stock market benchmark. VIX is computed using options traded on CBOE by solving the Black-Scholes formula for volatility and averaging over many options, \cite{Fear}. 

\subsection{Our use of VIX} In our companion article \cite{VIX-stocks}, we compute and analyze $Z(t)$. The VIX data is available daily from January 1986, if we include its predecessor for the S\&P 100 (a similar index for 100 stocks). Volatility therefore must be modeled as a separate time series/stochastic process, positive but mean-reverting. We can directly observe volatility each day measured by VIX, which is computed from many S\&P 500 index options traded on CBOE. The most natural model for VIX is autoregression of order 1 for the logarithm of VIX:~\eqref{eq:log-AR}. In continuous time, this would correspond to an Ornstein-Uhlenbeck process, see~\eqref{eq:log-OU}. We take the log to make VIX range over the entire real line, not just its positive half. The model fits well: $R^2$ is large, innovations (residuals) pass white noise tests. 

We stress the difference between this model, where volatility is directly observable, and classic stochastic volatility models, where only the return data is observable, and volatility is hidden and must be inferred from the available data. Classic stochastic volatility models were invented in 1980s and studied extensively; see for example articles \cite{Kastner, SV-MC, Stein}, monographs \cite{SV, TaylorBook}, and surveys \cite{VIX, Est-SV, TaylorReview}. These are alternatives for generalized autoregressive conditional heteroscedastic (GARCH) models. Stochastic volatility models (with hidden volatility) have two series of innovations instead of one in GARCH. We also mention the ARMA-GARCH models, where the observable data is modeled by ARMA, and the innovations are modeled using GARCH. The innovations therefore are weak but not strong white noise, see \cite{ARMA-GARCH1, ARMA-GARCH2, ARMA-GARCH3}.

We also mention the multivariate version of the Cox-Ingersoll-Ross (CIR) model, studied in \cite{Duffie-Kan}, for zero-coupon Treasury rates. The classic CIR model is given by the following SDE: 
$$
\mathrm{d}r(t) = (a - br(t))\,\mathrm{d}t + c\sqrt{r(t)}\,\mathrm{d}t,
$$
see \cite{CIR}. It attempts to capture rate movements and stochastic volatility in one equation. Thie article \cite{Duffie-Kan} has an extensive bibliography.

\subsection{The main idea} The main idea of this article is that VIX can be used as stochastic volatility for the Treasury bond market. In our approach, we need only to use classic ordinary least squares (OLS) linear regression. This was discussed in detail in our article \cite{VIX-stocks}. Similarly, the model~\eqref{eq:AR-SV},~\eqref{eq:log-AR} is an alternative to ARMA-GARCH models mentioned in the last subsection. Autoregressive stochastic volatility models have two series of innovations, as opposed to ARMA-GARCH models, which have only one series of innovations. Without observing volatility, autoregressive stochastic volatility models as in~\eqref{eq:AR-SV},~\eqref{eq:log-AR} are hard to estimate. But with volatility observed, estimation is only by a classic linear regression. 

We are surprised that stock market volatility, measured by VIX, works well not just for stock indices, as shown in our companion work \cite{VIX-stocks}, but for Treasury bonds as well. This is presumably because US stock and Treasury bond markets are closely related: They compete for investors. 

\begin{figure}[t]
\centering
\includegraphics[width=8cm]{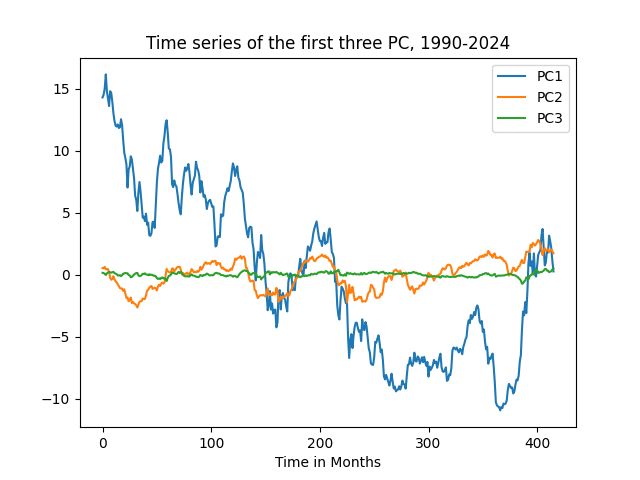}
\includegraphics[width=8cm]{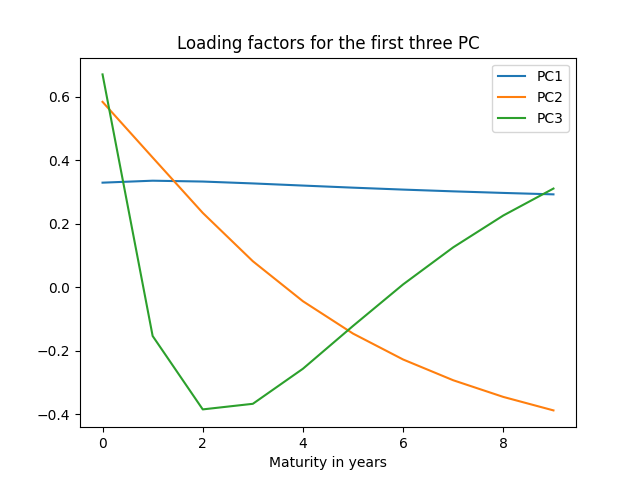}
\caption{Principal Components: level (PC1), slope (PC2), curvature (PC3), and their loading factors. In the left, we see the level gradually decreasing before sharply increasing at the end. This corresponds to the overall dynamics of rates since 1990. The slope and the curvature do not exhibit such clear-cut tendency. In the right, the loading factors are $c_{il}$.}
\label{fig:PCA}
\end{figure}

\subsection{Yield curve} The first step is to model bond rates; only after that we can model total returns. Dependence of a rate upon maturity is called the {\it yield curve}. Long-term rates are usually higher than short-term rates, but not always. In fact, {\it inverted yield curve} (when long rates are lower than short rates) is a reliable predictor of recessions. The {\it long-short spread} (the difference between long and short rates) is called the {\it slope} of the yield curve, and it exhibits mean reversion. 

Another factor, called {\it level}, captures the observation that all rates tend to rise and fall together. This factor also exhibits mean-reversion, since periods of high overall rates (like the 1980s) alternate with periods of very low rates (such as 2010s). Together, level and slope capture almost all variance. 

If one wants to study a tiny bit of remaining variance, one could consider {\it curvature}: The behavior of medium-term rates unexplained by movements of short- and long-term rates. See more discussion in \cite[Chapter 9, p. 276]{Ang} and \cite[Chapter 3]{Wu}.  Another approach is used in the classic book \cite{YieldBook}. There, the three components of the yield curve are modeled as deterministic functions. Statistical analysis is done using the expectation-maximization.

Usually, yield curve is given for classic coupon-paying Treasury bonds. However, here we use zero-coupon bonds which pay only principal at maturity. As discussed in Section 4 of this article, it makes total returns easier to computer for these bonds. In the Federal Reserve research report \cite{Kim-Wright}, a method was developed to compute zero-coupon Treasury rates for any given maturity. This is the data we use in this article. 

The Federal Reserve System controls short-term rates using its {\it open-market operations} (trading very short-term loans in the interbank market, the rate of which is almost perfectly correlated with, say, 3-month rate). Historically, long-term rates (5, 10, 30 years) are left to the open market, although more recently (in 2008-2009 and 2020 crises) the Fed intervened aggressively in this sector of the bond market as well (this program is popularly known as {\it quantitative easing}). See more in \cite[Chapter 9, subsection 2.1]{Ang}.

\subsection{Principal component analysis} Perform the principal components analysis (PCA) on $R_k(t)$ for $k = 1, \ldots, 10$. 
Denote the first three principal components {\it level, slope, curvature} (with their variance explained ratios):
\begin{align*}
P_1(t), & \quad P_2(t), \quad P_3(t),\\
96.63\% & \quad 3.31\% \quad 0.06\%.
\end{align*}
We see the time plots of these PC in the left panel of Figure~\ref{fig:PCA}, and the loading factors plot vs maturity in the right panel of Figure~\ref{fig:PCA}. One can see that well-nigh all variance is explained by the first principal component $P_1$, the {\it level}. The remaining variance is overwhelmingly explained by the second principal component $P_2$, the {\it slope}. The {\it curvature} $P_3$ explains the rest of the variance. We can see in the right panel of Figure~\ref{fig:PCA} that level loading factors are almost constant: This PC corresponds to the overall behavior of interest rates. The slope has decreasing loading factors: This PC shows how long-term rates behave relative to short-term rates. The curvature has a V-shaped pattern for its loading factors: It shows how medium-term rates behave other than following long-term and short-term rates. 

\subsection{Modeling volatility as autoregression} We follow the same technique as in our companion article \cite{VIX-stocks}. Fitting the VIX for monthly average VIX in 1990--2024 as in~\eqref{eq:log-AR}, we have: $\alpha = 0.34$ and $\beta = 0.88$. We test whether $\ln V(t)$ is a random walk using Augmented Dickey-Fuller test. The result is $p = 0.4\%$. Thus we reject the random walk null hypothesis. Judging by the Ljung-Box test for 10 lags, which gives $p = 50\%$ for $W$ and $p = 10\%$ for $|W|$, the innovations $W$ are IID. But skewness and kurtosis are $1.8$ and $10.7$, which makes it clear the data is not Gaussian. This confirms the results of our companion article. However, there are slight differences between this research and that companion article. Here, we fit a slightly smaller data series, starting from January 1990, not January 1986. We use only the VIX based on Standard \& Poor 500, not the version based on Standard \& Poor 100. 

\subsection{Scalar autoregression for each principal component} Fit a simple autoregression AR(1) for the first three principal components (level, slope, curvature):
\begin{equation}
\label{eq:AR}
P_k(t) = a_k + b_kP_{k-1} + Z_k(t),\quad k = 1, 2, 3.
\end{equation}
We have the following results:
\begin{align*}
a_1 &= -0.034,\quad b_1 = 1-0.012;\\
a_2 &= 2.8\cdot 10^{-3}, \quad b_2 = 1-0.016;\\
a_3 &= 7\cdot 10^{-4},\quad b_3 = 1-0.10.
\end{align*}
All values of $b_i$ are very close to $1$. This compels us to test the random walk hypothesis for $P_k$. The Dickey-Fuller test for each principal component, see \cite{DFtest}, gives us the following $p$-values: $p_1 = 18\%$, $p_2 = 3.5\%$ and $p_3 = 0.44\%$. Thus we fail to reject the random walk hypothesis for level, but we reject it for the slope and curvature. 

\begin{table}
\begin{tabular}{|c|c|c|c|c|c|c|}
\hline
Component & Level $P_1$ & Slope $P_2$ & Curvature $P_3$ \\
\hline
Skewness of $Z_k$ & 0.12 & -0.8 & -0.14 \\
\hline
Skewness of $Z_k/V$ & 0.24 & -0.37 & 0.08 \\
\hline
Kurtosis of $Z_k$ & 3.59 & 6.49 & 4.62 \\
\hline
Kurtosis of $Z_k/V$  & 3.92 & 4.01 & 4.58 \\
\hline
\end{tabular}
\caption{Skewness and kurtosis for autoregression innovations from~\eqref{eq:AR}. We see that division by $V$ makes $Z_2$ (innovations for slope) but not $Z_1$ and $Z_3$ (innovations for level and intercept) closer to Gaussian. Recall that skewness and kurtosis for the Gaussian law are 0 and 3.}
\label{table:norm}
\end{table}

For each $P_k$, we compute the skewness and kurtosis of innovations $Z_k(t)$ and of $Z_k(t)/V(t)$ from~\eqref{eq:AR}. The {\it empirical skewness} and {\it empirical kurtosis} of data $\varepsilon_1, \ldots, \varepsilon_N$ are defined as:
\begin{equation}
\label{eq:skewness-kurtosis}
\frac{\hat{m}_3}{s^3} \quad \mbox{and}\quad \frac{\hat{m}_4}{s^4},\quad \mbox{where}\quad  \hat{m}_k := \frac1N\sum\limits_{t=1}^N(\varepsilon_t - \overline{\varepsilon})^k.
\end{equation}
Here, $\overline{\varepsilon}$ and $s$ are the empirical mean and empirical standard deviation of this data. If kurtosis is closer to 3, and if skewness is closer to 0, this is a sign that the distribution is closer to Gaussian. One sees that dividing by VIX makes innovations closer to the normal distribution for the slope $P_2$, but not for the level $P_1$ and the curvature $P_3$, see Table~\ref{table:norm}. Figures~\ref{fig:qq} and~\ref{fig:acf} suggest that it is reasonable to model each innovations $Z_1$, $Z_2/V_2$, $Z_3$ as IID univariate Gaussian. 


\subsection{Motivation for our models} For cases when dividing innovations by VIX improves them, we take
\begin{equation}
\label{eq:corp}
P(t) = a + bP(t-1) + V(t)Z(t)
\end{equation}
using ordinary least squares regression after dividing by $V(t)$. But this would lack an intercept. We add a constant after normalization, which traslates into an additional term $cV(t)$ at the right-hand side of~\eqref{eq:corp}. We get:
\begin{equation}
\label{eq:ncorp}
P(t) = a + bP(t-1) + cV(t) + V(t)Z(t).
\end{equation}
Generalizing~\eqref{eq:ncorp} for many dimensions gives us~\eqref{eq:AR-SV}. We allow for the case when some innovations benefit from division by VIX but other innovations do not. The main result of consdering univariate models in~\eqref{eq:AR} is: Dividing innovations by VIX improves them (by making them closer to IID Gaussian) for the slope (PC2) but not for the level (PC1) and curvature (PC3). 

\subsection{Bivariate and trivariate models} Since curvature explains only a small part of the variance, let us model the first two principal components: level and slope. 
\begin{align}
\label{eq:bivariate}
\begin{split}
\ln V(t) &= \alpha + \beta \ln V(t-1) + Z_0(t);\\
P_1(t) &= a_1 + b_{11}P_1(t-1) + b_{12}P_1(t-1) + c_1V(t) + Z_1(t);\\
P_2(t) &= a_2 + b_{21}P_2(t-1) + b_{22}P_2(t-1) + c_2V(t) + V(t)Z_2(t).
\end{split}
\end{align}
We fit the second and third linear regressions in~\eqref{eq:bivariate} as follows:
\begin{align}
\label{eq:bivariate-fit}
\begin{split}
\mathbf{a} &= \begin{bmatrix} 0.3 & 0.079\end{bmatrix}\\
\mathbf{c} &= \begin{bmatrix} -0.0172 & -0.0039\end{bmatrix}\\
\mathbf{B} &= 
\begin{bmatrix}
1 - 0.0141 & 0.031\\ 
0.0005 & 1 - 0.0149
\end{bmatrix}
\end{split}
\end{align}
The $p$-values from the Student $t$-test show that we fail to reject hypotheses $b_{12} = 0$ and $b_{21} = 0$. Thus we might as well assume that the matrix $\mathbf{B}$ is diagonal.

\begin{figure}[t]
\centering
\subfloat[$Z_1$]{\includegraphics[width=5cm]{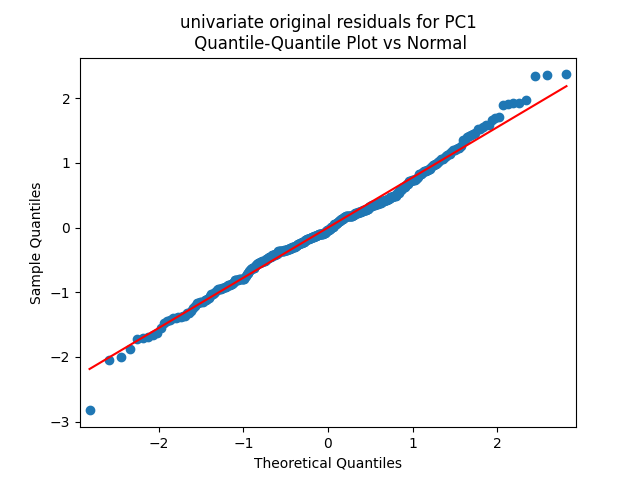}}
\subfloat[$Z_2$]{\includegraphics[width=5cm]{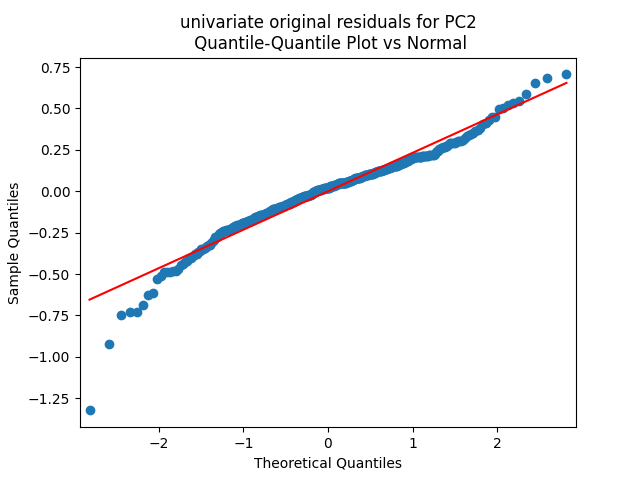}}
\subfloat[$Z_2/V$]{\includegraphics[width=5cm]{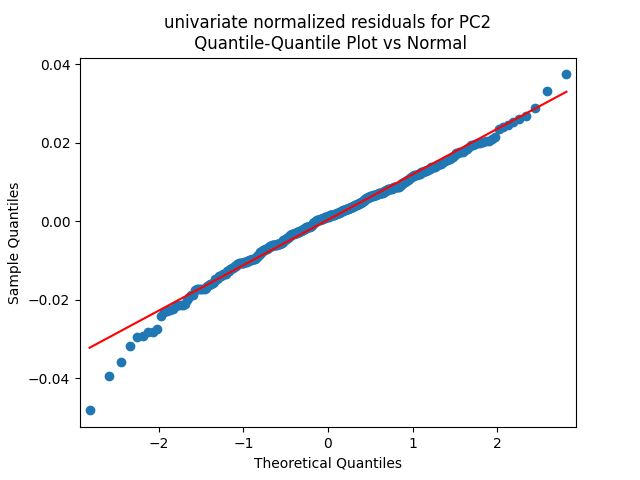}}
\caption{The quantile-quantile plots versus the Gaussian distributions for innovations $Z_1$ of PC1 (level); innovations $Z_2$ of PC2 (slope); and normalized innovations $Z_2/V$ of PC2, taken from~\eqref{eq:AR}. One sees that $Z_1$ is close to Gaussian, but $Z_2$ is not. However, normalizing $Z_2$ (dividing it by $V$) makes it closer to Gaussian.} 
\label{fig:qq}
\end{figure}

We can also generalize this for a version of vector autoregression with some but not all components having innovations with stochastic volatility. Let $\mathbf{P} = [P_1, P_2, P_3]$. For constant vectors $\mathbf{a}, \mathbf{c} \in \mathbb R^3$ and a constant $3\times 3$-matrix $\mathbf{B}$,
\begin{align}
\label{eq:trivariate}
\begin{split}
\ln V(t) &= \alpha + \beta \ln V(t-1) + Z_0(t);\\
\mathbf{P}(t) &= \mathbf{a} + \mathbf{B}\mathbf{P}(t-1) + \mathbf{c}V(t) + [Z_1(t), V(t)Z_2(t), Z_3(t)].
\end{split}
\end{align}
In the model~\eqref{eq:trivariate}, we can fit linear regression for $P_1$ and $P_3$. To fit the one for $P_2$, we divide the equation by $V$, and fit after that:
\begin{align}
\label{eq:trivariate-fit}
\begin{split}
\mathbf{a} &= [0.2844, 0.0667, -0.0054],\\
\mathbf{c} &= [-0.0164, -0.0033, 0.0003],\\
\mathbf{B} &= \begin{bmatrix}1 - 0.0140 & 0.0310 & -0.3881 \\ 0.0005 & 1 - 0.0109 & -0.2174 \\ 0.001 & 0.0061 & 1 - 0.0986\end{bmatrix}
\end{split}
\end{align}
Let us discuss the fit~\eqref{eq:trivariate-fit}. The $p$-values for the Student $t$-test in each regression for $\triangle P_k(t)$ for the factor $P_l(t-1)$ if $k \ne l$ are greater than $5\%$, except for $k = 3$ and $l = 2$; then $p = 4.8\%$. Thus we fail to reject (almost) all cross-dependence of $\triangle P_k(t)$ upon $P_l(t-1)$ for $k \ne l$. In sum, we might assume that $\mathbf{B}$ is diagonal. Also, the $p$-value for $c_k = 0$ is $p_1 = 0.1\%$, $p_2 = 8.4\%$, $p_3 = 50\%$, so we reject $c_1 = 0$, but fail to reject $c_2 = 0$ and $c_3 = 0$. 

For innovations $Z_k$, the plots of the autocorrelation function (ACF) and the quantile-quantile (QQ) plot versus the normal distribution show that these can be reasonably described as IID Gaussian. Other plots were generated in the code and are not included here for lack of space; see the \texttt{GitHub} repository \texttt{asarantsev/vix-zeros} and the Python code. 

\begin{figure}[t]
\centering
\subfloat[$Z_1$]{\includegraphics[width=5cm]{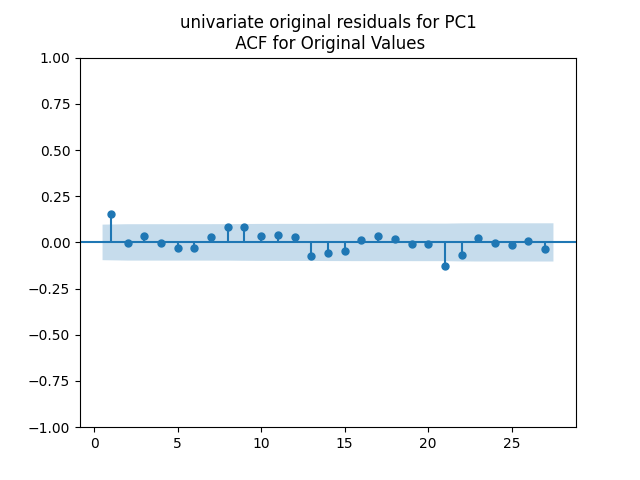}}
\subfloat[$Z_2/V$]{\includegraphics[width=5cm]{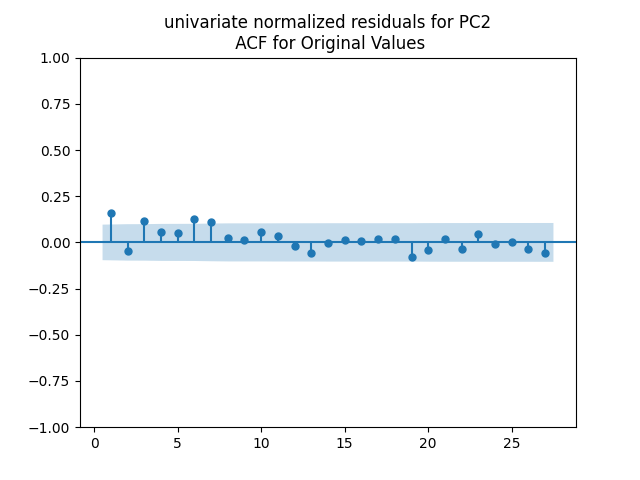}}
\subfloat[$Z_3$]{\includegraphics[width=5cm]{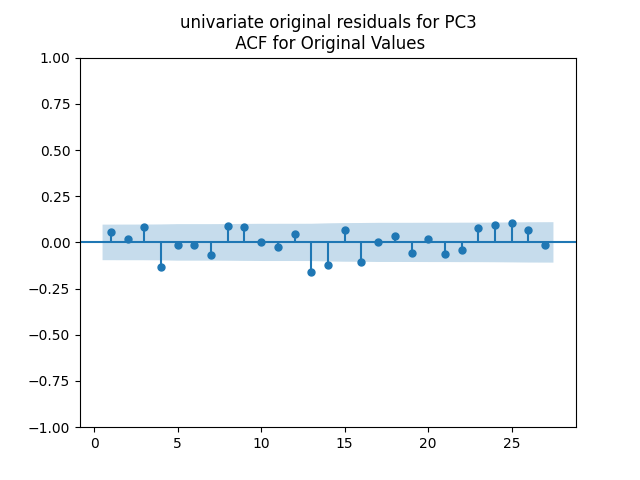}}
\caption{Autocorrelation functions for innovations $Z_1$ of PC1 (level); normalized innovations $Z_2/V$ of PC2 (slope); and innovations $Z_3$ of PC3 (curvature), taken from~\eqref{eq:AR}. We see that some values of the ACF are high enough to be outside of the shadowed area. If this value is indeed outside, then applying the white noise test for this lag value, we reject the white noise hypothesis. However, most ACF values in each case are small enough. We consider it reasonable to model $Z_1, Z_2, Z_3$ using IID Gaussian.} 
\label{fig:acf}
\end{figure}

The bivariate and trivariate models fall under the conditions of Theorem~\ref{thm:stationary}, because for both~\eqref{eq:bivariate-fit} and~\eqref{eq:trivariate-fit} have real eigenvalues between 0 and 1: $0.93, 0.97, 0.98$ for the model~\eqref{eq:trivariate}, and $0.98, 0.99$ for the model~\eqref{eq:bivariate}. 

\section{Long-Term Stability Results}

\subsection{Statements and proofs for discrete time} Fix dimension $d \ge 1$ and another dimension $r = 0, \ldots, d$. Consider the model~\eqref{eq:AR-SV},~\eqref{eq:log-AR},~\eqref{eq:diag}.

\begin{theorem}
Assume $\beta \in (0, 1)$ and all eigenvalues of $\mathbf{B}$ have absolute value less than 1. In addition, assume that 
\begin{equation}
\label{eq:MGF-W}
\mathbb E\left[e^{uZ_0(n)}\right] < \infty,\quad |u| \le u_0;
\end{equation}
and for each $i = 1, \ldots, d$, there exists a $u_i > 0$ such that $\mathbb E[|Z_i(n)|^{u_i}] < \infty$. Then the process $(\ln V, \mathbf{X})$ has a stationary distribution in $\mathbb R^{d+1}$. 
\label{thm:stationary}
\end{theorem}

\begin{proof}  {\it Step 1.} We rewrite the model equation~\eqref{eq:AR-SV} as follows:
\begin{align}
\label{eq:rewrite}
\begin{split}
\mathbf{X}(n) &= A_n\mathbf{X}(n-1) + B_n;\\
A_n &:= \mathbf{B},\quad B_n := \mathbf{a} + \mathbf{c}V(n) + \xi(n)\mathbf{Z}(n).
\end{split}
\end{align}
The rest of the proof relies heavily on the main result of \cite{Brandt}, which is rewritten as \cite[Theorem 1.1]{Bougerol} for the multivariate setting and in a more general way. We need to show the following:
\begin{equation}
\label{eq:A-n}
\mathbb E\max(\ln|\!|A_n|\!|_2, 0) < \infty;
\end{equation}
\begin{equation}
\label{eq:B-n}
\mathbb E\max(\ln|B_n|_2, 0) < \infty;
\end{equation}
\begin{equation}
\label{eq:Lyapunov}
\inf_nn^{-1}\cdot \mathbb E\left[\ln|\!|A_1\cdot \ldots\cdot A_n|\!|_2\right] < 0.
\end{equation}

{\it Step 2.} Condition~\eqref{eq:A-n} follows immediately from the definition of $A_n$ in~\eqref{eq:rewrite}, since the matrix $\mathbf{B}$ is deterministic. 

{\it Step 3.} Condition~\eqref{eq:B-n} can be proved as follows: for any $\varepsilon > 0$, there exists a constant $C(\varepsilon)$ such that $\max(\ln u, 0) \le C(\varepsilon)u^{\varepsilon}$ for $u \ge 0$ (assuming $\ln 0 := -\infty$). Therefore, it suffices to show that $B_n \in L^{\varepsilon}$; in other words, 
\begin{equation}
\label{eq:B-n-u}
\mathbb E[|B_n|_2^{\varepsilon}] < \infty.
\end{equation}
Apply the triangle inequality: 
\begin{equation}
\label{eq:triangle-Bn}
|B_n|_2 \le |\mathbf{a}|_2 + |\mathbf{c}|_2V(n) + |\!|\xi(n)|\!|_2\cdot|\mathbf{Z}(n)|_2.
\end{equation}
The spectral norm of a diagonal matrix equals the maximal absolute value of its elements, in our case 
\begin{equation}
\label{eq:diag-xi}
|\!|\xi(n)|\!|_2 = \max(V(n), 1) \le V(n) + 1.
\end{equation}
In light of~\eqref{eq:triangle-Bn},~\eqref{eq:diag-xi}, to show~\eqref{eq:B-n-u}, we need 
$|\mathbf{a}|_2 + |\mathbf{c}|_2V(n) + (V(n)+1)\cdot|\mathbf{Z}(n)|_2 \in L^{\varepsilon}$. Pick  
\begin{equation}
\label{eq:eps-u}
\varepsilon = 0.5\min(u_0, u_1, \ldots, u_d),
\end{equation}
since by our companion article \cite[Lemma 1]{VIX-stocks}, $\mathbb E\left[V^{u_0}(n)\right] < \infty$; and from the assumptions we have: $\mathbb E\left[|Z_i(n)|^{2\varepsilon}\right] < \infty$ for all $i = 1, \ldots, d$; therefore, 
$$
\mathbb E\left[|\mathbf{Z}(n)|_2^{2\varepsilon}\right] = \mathbb E\left[Z_1^2(n) + \ldots + Z_d^2(n)\right]^{\varepsilon} < \infty. 
$$
We used the observation that a finite sum of random variables with finite moment of order $\delta$ also has finite moment of order $\delta$. 

{\it Step 4.} Condition~\eqref{eq:Lyapunov} follows from the Gelfand formula, see \cite[Corollary 5.6.14]{Matrix}:
\begin{equation}
\label{eq:Gelfand}
\lim\limits_{n \to \infty}n^{-1}\ln|\!|\mathbf{B}^n|\!|_2 = \ln\left(\max\{|\lambda| \mid \lambda \in \sigma(\mathbf{B})\}\right) < 0.
\end{equation}
Since we showed~\eqref{eq:A-n},~\eqref{eq:B-n},~\eqref{eq:Lyapunov}, we complete the proof. 
\end{proof}

\begin{theorem} Under assumptions of Theorem~\ref{thm:stationary}, if $\mathbf{W}$ has Lebesgue density which is everywhere positive in $\mathbb R^{d+1}$, then this process is ergodic. 
\label{thm:ergodicity}
\end{theorem}

\begin{proof}
This process has a stationary distribution, as shown in Theorem~\ref{thm:stationary}. From positivity of Lebesgue density of $\mathbf{W}$ it is straightforward to show that the transition density of this Markov process is everywhere positive as well. Apply Lemma~\ref{lemma:ergodic-discrete} from the Appendix. 
\end{proof}

\begin{theorem}
Under assumptions of Theorem~\ref{thm:ergodicity}, if $\mathbb E[e^{2Z_0(n)}] < \infty$ and $\mathbb E[Z^2_i(n)] < \infty$ for each $i = 1, \ldots, d$, then the Strong Law of Large Numbers is satisfied: As $N \to \infty$, 
\begin{equation}
\label{eq:LLN}
\frac{\mathbf{X}(1) + \ldots + \mathbf{X}(N)}{N} \to \int_{\mathbb R^d}\mathbf{x}\,\pi(\mathrm{d}\mathbf{x})\quad \mbox{almost surely,}
\end{equation}
where $\pi$ is the marginal of the stationary distribution for $(\ln V, \mathbf{X})$ corresponding to $\mathbf{X}$. 
\label{thm:LLN}
\end{theorem}

\begin{proof}
Using \cite[Theorem 17.0.1(i)]{Stability}, we need only to prove that the right-hand side of~\eqref{eq:LLN} is finite. According to the main result of \cite{Brandt} and \cite[Theorem 1.1]{Bougerol}, the stationary probability measure $\pi$ is the distribution of the random variable
\begin{equation}
\label{eq:X-infty}
\mathbf{X}(\infty) = \sum\limits_{n=0}^{\infty}\mathbf{B}^n(\mathbf{a} + \mathbf{c}V(-n) + \xi(-n)\mathbf{Z}(-n)).
\end{equation}
We extend sequence of innovations $(W(n), \mathbf{Z}(n))$ for integer $n \le 0$. Take $|\cdot|_2$-norm in~\eqref{eq:X-infty}: 
\begin{equation}
\label{eq:norm-infty}
|\mathbf{X}(\infty)|_2 \le \sum\limits_{n=0}^{\infty}|\!|\mathbf{B}|\!|_2^n|\mathbf{a} + \mathbf{c}V(-n) + \xi(-n)\mathbf{Z}(-n)|_2.
\end{equation}
Use~\eqref{eq:B-n-u} from the proof of Theorem~\ref{thm:stationary}. Apply this to $u_0 = \ldots = u_d = 2$. From~\eqref{eq:eps-u}, we get $\varepsilon = 1$. Thus we show that the constant $\mathbb E\left[|B_n|_2\right] < \infty$. From stationarity of $B_n$, it follows that this absolute moment is independent of $n$. Therefore, the series in the right-hand side of~\eqref{eq:norm-infty} converges. This completes the proof.
\end{proof}

\subsection{Statements for continuous time} Our discrete-time models are Markov: the value at step $t$ depends only on the value at step $t-1$. Therefore, they can be readily generalized to the continuous-time setting. The innovation terms $W_i(t)$ become a Brownian motion with zero drift but a nontrivial constant diffusion coefficient. An autoregression of order 1 becomes an Ornstein-Uhlenbeck process for $\ln V(t)$, as i~\eqref{eq:log-OU}, see the Introduction. For the level $P_1(t)$ in continuous time:
$$
\mathrm{d}P_1(t) = (a_1 - b_1P_1(t) + c_1V(t))\,\mathrm{d}t + \mathrm{d}W_1(t).
$$
The equation for the slope $P_2(t)$ becomes
$$
\mathrm{d}P_2(t) = (a_2 - b_2P_2(t) + c_2V(t))\,\mathrm{d}t + V(t)\,\mathrm{d}W_2(t).
$$
We can also generalize this to the multivariate case and get~\eqref{eq:diag},~\eqref{eq:OU-SV},~\eqref{eq:log-OU}. We can rewrite this for $\mathbf{X}(t) = (X_1(t), \ldots, X_d(t))$ as:
\begin{align}
\label{eq:SDE}
\begin{split}
\mathrm{d}\mathbf{X}(t) &= (\mathbf{a}-\mathbf{B}\mathbf{X}(t)+\mathbf{c}V(t))\,\mathrm{d}t + \mathrm{d}\mathbf{M}(t),\\
\mathbf{M}(t) &= (M_1(t), \ldots, M_d(t));\\
\mathrm{d}M_i(t) &= V(t)\,\mathrm{d}W_i(t),\quad i = 1, \ldots, r;\\
\mathrm{d}M_i(t) &= \,\mathrm{d}W_i(t),\quad i = r+1, \ldots, d,
\end{split}
\end{align}
where $\mathbf{W} = (W_0, \ldots, W_d)$ is a $d+1$-dimensional Brownian motion with mean vector zero and covariance matrix $\Sigma$. 

\begin{asmp} Here, $\mathbf{a}, \mathbf{c} \in \mathbb R^d$ are constant vectors, $\mathbf{B}$ is a $d\times d$ constant matrix with all eigenvalues having positive real part.
\label{asmp:const}
\end{asmp}

We can generalize the dynamics of the local $\mathbb R^d$-valued martingale $\mathbf{M} = (\mathbf{M}(t),\, t \ge 0)$.

\begin{asmp} For each $i$, the $i$th component is a local martingale ${M}_i$ with quadratic variation $\langle M_i\rangle$ which satisfies $\mathrm{d}\langle M_i\rangle_t = Q_{i}(V(t))\,\mathrm{d}t$, where $Q_i : (0, \infty) \to [0, \infty)$, and for some  constant $c_i$ we have: $Q_i(v) \le c_i(v^2+1)$ for all $v > 0$. 
\label{asmp:mgle}
\end{asmp} 

Lastly, we set an assumption on the Ornstein-Uhlenbeck process for $\ln V$ in~\eqref{eq:log-OU}.

\begin{asmp}
In~\eqref{eq:log-OU}, we have: $\beta > 0$.
\label{asmp:beta}
\end{asmp}

\begin{theorem} Consider the combined $d+1$-dimensional model~\eqref{eq:log-OU} and~\eqref{eq:SDE}: $(\ln V(t), \mathbf{X}(t))$. Under Assumptions~\ref{asmp:const},~\ref{asmp:mgle},~\ref{asmp:beta}, the process is ergodic, and satisfies the Strong Law of Large Numbers: For any measurable function $f : \mathbb R^d \to \mathbb R$ such that $|f(\mathbf{x})| \le |\mathbf{x}|_2^2$ for all $\mathbf{x} \in \mathbb R^d$, 
\begin{equation}
\label{eq:SLLN-ct}
\frac1T\int_0^Tf(\mathbf{X}(t))\,\mathrm{d}t \to \int_{\mathbb R^d}f\mathrm{d}\pi,\quad T \to \infty,\quad \mbox{a.s.}
\end{equation}
where $\pi$ is the marginal of the stationary distribution for $(\ln V, \mathbf{X})$ corresponding to $\mathbf{X}$.
\label{thm:cont-time} 
\end{theorem}

In particular, consider the model~\eqref{eq:diag},~\eqref{eq:OU-SV},~\eqref{eq:log-OU}: multivariate Ornstein-Uhlenbeck process with stochastic volatility. Under Assumptions~\ref{asmp:const} and~\ref{asmp:beta}, the process is ergodic, and satisfies the Strong Law of Large Numbers~\eqref{eq:SLLN-ct}. 

\subsection{Proof of Theorem~\ref{thm:cont-time}} {\it Step 1.} Under Assumption~\ref{asmp:beta}, for any $u \in \mathbb R$, 
\begin{equation}
\label{eq:bdd-moment}
\sup_{t \ge 0}\mathbb E[V^u(t)] < \infty.
\end{equation}
Indeed, the process $\ln V(\cdot)$ is an Ornstein-Uhlenbeck process, with stationary distribution $\mathcal N(\mu_{\infty}, \rho^2_{\infty})$. For any $t \ge 0$, the distribution of $\ln V(t)$ is also Gaussian: $\mathcal N(\mu_t, \sigma^2_t)$, where $\mu_t \to \mu_{\infty}$, $\sigma_t \to \sigma_{\infty}$, $t \to \infty$. This is taken from \cite[p.106]{Handbook} or \cite[p.358]{KSBook}. Using the moment generating function of a normal distribution:
\begin{equation}
\label{eq:V-conv}
\mathbb E[V^u(t)] = \mathbb E[\exp(u\ln V(t))] = \exp\left(u\mu_t + 0.5\sigma_t^2u^2\right) \to \exp\left(u\mu_{\infty} + 0.5\sigma_{\infty}^2u^2)\right). 
\end{equation}

{\it Step 2.} We can explicitly solve this linear multivariate stochastic differential equation with respect to $X(t)$ if we assume we already know $V(t)$. We use the matrix exponent $e^{\mathbf{B}t}$. Its properties can be found in any standard textbook, for example \cite{Hall}. See also solution of linear stochastic differential equations in \cite[Chapter 5, Section 6]{KSBook}. We can write 
\begin{equation}
\label{eq:complete}
\mathbf{X}(t) = e^{-\mathbf{B}t}\mathbf{X}(0) + \int_0^te^{\mathbf{B}(s-t)}(\mathbf{a} + \mathbf{c}V(t))\,\mathrm{d}s + \int_0^te^{\mathbf{B}(s-t)}\,\mathrm{d}\mathbf{M}(s).
\end{equation}
We can write~\eqref{eq:complete} as the sum of three terms:
\begin{align}
\label{eq:decomp}
\begin{split}
\mathbf{X}(t) &:= \mathcal A_1(t) + \mathcal A_2(t) +\mathcal A_3(t);\\
\mathcal A_1(t) &:= e^{-\mathbf{B}t}\mathbf{X}(0) + \int_0^te^{\mathbf{B}(s-t)}\mathbf{a}\,\mathrm{d}s;\\
\mathcal A_2(t) &:= \int_0^te^{\mathbf{B}(s-t)}\mathbf{c} V(s)\,\mathrm{d}s;\\
\mathcal A_3(t) &:= \int_0^te^{\mathbf{B}(s-t)}\,\mathrm{d}\mathbf{M}(s).
\end{split}
\end{align}

{\it Step 3.} By \cite[Theorem 2, Section 1.9]{ODE}, there exist constants $k_0, k_1 > 0$ such that
\begin{equation}
\label{eq:exp-ineq}
|\!|e^{-\mathbf{B}u}|\!|_2 \le k_1e^{-k_0u},\quad u \ge 0.
\end{equation}
Therefore, the $\mathcal A_1(t)$ from~\eqref{eq:decomp} is bounded almost surely by a certain constant $C_1$:
\begin{equation}
\label{eq:A1}
\sup\limits_{t \ge 0}|\mathcal A_1(t)|_2 \le C_1.
\end{equation}

{\it Step 4.} Next, we can estimate the norm of the random function with finite variance from the right-hand side of~\eqref{eq:complete}:
Indeed, apply this norm and use~\eqref{eq:exp-ineq}:
\begin{equation}
\label{eq:A-2}
|\mathcal A_2(t)|_2 \le \int_0^t|\!|e^{\mathbf{B}(s-t)}|\!|_2\cdot |\mathbf{c}|_2 V(s)\,\mathrm{d}s \le \int_0^tk_1|\mathbf{c}|_2e^{k_0(s-t)}V(s)\,\mathrm{d}s.
\end{equation}
Apply the expected value to~\eqref{eq:A-2}: The right-hand side of~\eqref{eq:A-2} is bounded over $t \ge 0$. Thus 
\begin{equation}
\label{eq:A2}
\sup\limits_{t \ge 0}\mathbb E|\mathcal A_2(t)|^2_2 < \infty.
\end{equation}

{\it Step 5.} Finally, let us estimate the second moment of the local martingale part $\mathcal A_3(t)$ from~\eqref{eq:decomp}. We use Lemma~\ref{lemma:matrix-mgle} with $\mathcal G_i(s) = Q_i(V(s))$ and $F(s) = e^{\mathbf{B}(s-t)}$. The martingale $\mathcal M$ has square-integrable components, and from Assumption~\ref{asmp:mgle} and results of Step 1 of this proof: Using estimates $\mathcal G_i(s) \le a_0V^2(s) + a_1$, we get:
\begin{equation}
\label{eq:A3}
\mathbb E\left[|\mathcal A_3(t)|_2^2\right] \le d^{3/2}\int_0^tk_1e^{2k_0(s-t)}\left(a_0k_V + a_1\right)\,\mathrm{d}s \le k_1(a_0k_V + a_1)\int_0^{\infty}e^{-2k_0s}\,\mathrm{d}s.
\end{equation}
The right-hand side of~\eqref{eq:A3} is bounded over $t \ge 0$. 

{\it Step 6.} Combining~\eqref{eq:A1},~\eqref{eq:A2},~\eqref{eq:A3}, 
we get that $\mathbf{X} := \mathcal A_1 + \mathcal A_2 + \mathcal A_3$ is tight, since 
\begin{equation}
\label{eq:finite-second}
\sup\limits_{t \ge 0}\mathbb E\left[|\mathbf{X}(t)|^2_2\right] < \infty.
\end{equation}
Using~\eqref{eq:bdd-moment} with $u = \pm 2$ and~\eqref{eq:finite-second}, we see that the process $(\ln V, \mathbf{X})$ is tight. From Lemma~\ref{lemma:ergodic-continuous} in the Appendix, we complete the proof of ergodicity. 

{\it Step 7.} To show the Strong Law of Large Numbers, we need to apply Lemma~\ref{lemma:SLLN} in the Appendix to $U(v, \mathbf{x}) = e^{2|v|} + |\mathbf{x}|_2^2$. To check~\eqref{eq:BDD}, we apply again~\eqref{eq:bdd-moment} with $u = \pm 2$ and~\eqref{eq:finite-second}. Therefore, for functions $f : \mathbb R^{d+1} \to \mathbb R$ such that $|f(v, \mathbf{x})| \le U(v, \mathbf{x})$ for all $v \in \mathbb R$ and $\mathbf{x} \in \mathbb R^d$. In particular, consider functions $f : \mathbb R^d \to \mathbb R$ independent of $v$ (that is, $f$ does not contain the VIX). Then we only need $|f(\mathbf{x})| \le |\mathbf{x}|_2^2$ for all $\mathbf{x} \in \mathbb R^d$.

\section{Total Zero-Coupon Bond Returns}

\subsection{Yield to maturity and coupons} We choose zero-coupon Treasury bonds (instead of the classic coupon-paying bonds). They pay only principal at maturity, and do not pay any coupons. This is in stark contrast with classic Treasury bonds which pay semiannual coupons, and return the principal at maturity. Let us explain the motivation: It is much easier to compute total returns for zero-coupon Treasury bonds. 

Unlike even the safest corporate bonds, Treasury bonds do not have any {\it default risk} (missing or not paying in full the coupon and principal payments). However, Treasury bonds do have {\it interest rate risk:} When the overall level of interest rates rises, bond prices fall to make these bonds competitive with newly issued bonds at these higher rates. To quantify this, the following {\it yield to maturity} (YTM) concept was developed, which summarizes relationships between current bond price, future coupon and principal payments. This is the rate $r$ of return which we get if we reinvest all payments at this same rate: the only solution of 
\begin{equation}
\label{eq:YTM}
P_0 = \sum\limits_{k=1}^n\frac{P_k}{(1+r)^{t_k}}
\end{equation}
where $P_0$ is the current price (at time $t = 0$), and $P_k$ is a future payment at time $t_k$. This value $r$ is quoted most frequently in financial datasets and popular media, and is commonly known as the {\it bond rate}. The equation~\eqref{eq:YTM} gives us dependence of $P_0$ upon $r$: the derivative $\partial P_0/\partial r$ is negative, and its absolute value is called the {\it duration} of this bond. 

\subsection{Motivation for using zero-coupon bonds} A naive belief that, say, a 4\% Treasury bond returns 4\% over the next year is wrong. The total returns are composed of coupon/principal payments and price changes. More precisely, if $S(t)$ is the price at end-of-month $t$ and $D(t)$ is the payments during month $t$, then total (log) returns this month are computed as
\begin{equation}
\label{eq:TR}
Q(t) = \ln\frac{S(t)+D(t)}{S(t-1)}.
\end{equation}
The price changes, as follows from~\eqref{eq:YTM}, depend on the rate changes. It is not immediately clear how to model these dependencies using simple equations. We could take an exchange-traded fund (ETF) of Treasury bonds (or of a certain sector, for example long or short Treasuries) and model its total returns as a linear regression upon past rates and rate changes. In this article, we use an alternative: Zero-coupon Treasury bonds ({\it zeros}), which pay only principal $P$ at maturity $T$ and no coupons. The equation~\eqref{eq:YTM} then much simplifies:
\begin{equation}
\label{eq:zeros}
P_0 = \frac{P}{(1+r)^T}.
\end{equation}
Then $D(t) = 0$, and~\eqref{eq:TR} simplifies to $Q(t) = \ln S(t) - \ln S(t-1)$. Using~\eqref{eq:zeros} makes total returns much easier to compute for zero-coupon Treasury bonds (zeros). One can infer their prices using the Duffie-Kan three-factor model. This was done in \cite{Kim-Wright} and the resulting data is available daily from January 1990 for zeros of maturities from 1 to 10 years. 

\subsection{Bond returns in discrete time} Consider a zero-coupon rate $R_k$ corresponding to $k$-year maturity. Denote by $c_{ik}$ the loading factor for $R_k$ corresponding to the $i$th principal component (PC). The first three PC capture almost all variance, so we can approximate 
$$
R_k(t) = c_{1k}P_1(t) + c_{2k}P_2(t) + c_{3k}P_3(t).
$$
However, for the sake of generality, assume we have $d$ PC:
$$
R_k(t) = \sum\limits_{i=1}^{d}c_{ik}P_i(t).
$$
We have monthly data $t$ but annual maturity. That is, we have data on zero-coupon bonds for January 1990, February 1990, \ldots, but only for maturities of 12 months, 24 months, \ldots, and not for 13 months or 18 months. To carry out analysis, we must have data for fractional-year maturity. The easiest way is to interpolate $c_{ik}$ for monthly $k$ and define new rates with fractional-year maturity $\rho_l(t)$ for $l$ months. These loading factors with their linear interpolations are already shown in Figure~\ref{fig:PCA}.  Then $P_k(t) = \rho_{12k}(t)$. We let
\begin{equation}
\label{eq:PC-rho}
\rho_l(t) = \sum\limits_{i=1}^{d}\gamma_{il}P_i(t),\quad \gamma_{i, 12l} = c_{il}.
\end{equation}
From~\eqref{eq:zeros}, the price of a zero-coupon bond of maturity $l$ at end-of-month $t-1$ is $S_l(t-1) = (1 + \rho_l(t-1))^{-l}$. Next month, at end-of-month $t$ is $S_l(t) = (1 + \rho_l(t))^{-l+1}$. This bond does not pay any coupons, thus its total returns are equal to price returns, which are:
$$
Q_l(t) = \ln\frac{S_l(t)}{S_l(t-1)} = -(l-1)\ln(1 + \rho_{l-1}(t)) + l\ln(1+\rho_l(t-1)).
$$
We can approximate $\ln(1+x) \approx x$ for small $x \approx 0$. Since all rates are small (between 0 and 10\%), the returns are approximately 
\begin{equation}
\label{eq:approx-returns}
Q_l(t) \approx Q_l^*(t) = -(l-1)\rho_{l-1}(t) + l\rho_l(t-1).
\end{equation}
Plugging the representation~\eqref{eq:PC-rho} into~\eqref{eq:approx-returns}, and letting $\Gamma_{il} = l\gamma_{il}$, we get:
\begin{align}
\label{eq:approximation}
\begin{split}
Q_l^*(t) := \sum\limits_{i=1}^{d}\left[-\Gamma_{i,l-1}P_i(t-1) + \Gamma_{il}P_i(t)\right],\quad l = 1, \ldots, d.
\end{split}
\end{align}
Define by $(P_1(\infty), \ldots, P_d(\infty))$ random variables having the stationary distribution which exists by Theorem~\ref{thm:cont-time}. This linear representation~\eqref{eq:approximation} allows to quickly  prove the Strong Law of Large Numbers for zero-coupon Treasury bond returns. 

\begin{theorem} Assume $(P_1, \ldots, P_d) = \mathbf{X}$ and $V$ are modeled as~\eqref{eq:AR-SV},~\eqref{eq:log-AR},~\eqref{eq:diag}. Under Assumptions ~\ref{asmp:const},~\ref{asmp:mgle},~\ref{asmp:beta}, total returns $Q^*_l$ for maturity $l$ in~\eqref{eq:approximation} satisfy 
\begin{equation}
\label{eq:conv}
\frac{Q^*_l(1) + \ldots + Q^*_l(T)}{T} \to \sum\limits_{i=1}^d(\Gamma_{il} - \Gamma_{i, l-1})\cdot \mathbb E[P_i(\infty)], \quad T \to \infty\quad \mbox{almost surely.}
\end{equation}
\label{thm:returns-discrete}
\end{theorem}

\begin{proof}
From~\eqref{eq:approximation}, we get:
$$
\frac{Q^*_l(1) + \ldots + Q^*_l(T)}{T} = \sum\limits_{i=1}^d\Gamma_{il}\cdot\frac1T\sum\limits_{t=1}^TP_i(t) - \sum\limits_{i=1}^d\Gamma_{i,l-1}\cdot\frac1T\sum\limits_{t=1}^TP_i(t-1).
$$
The rest follows from Theorem~\ref{thm:LLN}. 
\end{proof}

\subsection{Bond returns in continuous time} We replace total returns $Q^*_l(t)$ of the zero-coupon Treasury bonds in~\eqref{eq:approximation} with maturity $l$ by the It\^o differential $\mathrm{d}\ln\mathcal W_l(t)$. Here, $\mathcal W_l = (\mathcal W_l(t),\, t \ge 0)$ is the {\it wealth process:} amount of wealth at time $t$ if $\mathcal W_l(0) = 1$ is invested in zero-coupon Treasury bonds with maturity $l$. Our goal is to write a stochastic equation for $\mathcal W_l$. We rewrite the right-hand side of~\eqref{eq:approximation} as
$$
\sum\limits_{i=1}^d\Gamma_{il}(P_i(t) - P_i(t-1)) + \sum\limits_{i=1}^d(\Gamma_{il} - \Gamma_{i, l-1})P_i(t-1).
$$
Switching to continuous time, we change $\Gamma_{il}$ to $\Gamma_i(l)$, because now $l$ is a continuous positive parameter. Next, we change $\Gamma_{il} - \Gamma_{i, l-1}$ to $\Gamma_i'(l)$ (derivative with respect to $l$). Further, we switch from $P_i(t)$ to $P_i(t)\,\mathrm{d}t$, and from $P_i(t) - P_i(t-1)$ to $P_i(t)\,\mathrm{d}t$. All in all, we rewrite \eqref{eq:approximation} in continuous time as
\begin{equation}
\label{eq:ct-returns}
\mathrm{d}\ln \mathcal W_l(t) = \sum\limits_{i=1}^d\Gamma_i'(l)P_i(t)\,\mathrm{d}t + \sum\limits_{i=1}^d\Gamma_i(l)\mathrm{d}P_i(t).
\end{equation}

\begin{asmp}
\label{asmp:zero}
Each deterministic function $\Gamma_i : [0, \infty) \to \mathbb R$ is continuously differentiable. 
\end{asmp}

\begin{theorem} Assume $(P_1, \ldots, P_d) = \mathbf{X}$ and $V$ are modeled as~\eqref{eq:diag},~\eqref{eq:OU-SV},~\eqref{eq:log-OU}. Under Assumptions ~\ref{asmp:const},~\ref{asmp:mgle},~\ref{asmp:beta},~\ref{asmp:zero}, as $T \to \infty$, we have convergence in probability
$$
\frac1T\ln\mathcal W_l(T) \to \sum\limits_{i=1}^p\Gamma'_i(l)\cdot\mathbb E[P_i(\infty)],
$$
where $P_i(\infty)$ corresponds to the stationary distribution.
\label{thm:returns-continuous} 
\end{theorem}

\begin{proof}
Apply Theorem~\ref{thm:cont-time} to $f(\mathbf{x}) = x_i$ for each $i = 1, \ldots, d$. Note that $\Gamma_i(l)$ and $\Gamma_i'(l)$ are simply deterministic constants. Therefore, almost surely (and in probability)
$$
\frac1T\int_0^TP_i(t)\,\mathrm{d}t \to \mathbb E[P_i(\infty)],\quad T \to \infty.
$$
Next, we need to show another convergence in probability:
$$
\frac1T\int_0^T\,\mathrm{d}P_i(t) \equiv \frac{P_i(T) - P_i(0)}{T} \to 0,\quad T \to \infty.
$$
But this follows from~\eqref{eq:finite-second} and the classic Markov inequality, \cite[Proposition 5.1.1]{Look}. 
\end{proof}

\subsection{Term premia and CAPM} Assume for once that we have only one principal components: level. This is reasonable, since it explains most of the variance, see subsection 2.7. Let us operate in continuous time. (We leave the reader to check we have similar results for discrete time.) Rewrite~\eqref{eq:ct-returns} as
\begin{equation}
\label{eq:ct-ret-pc1}
\mathrm{d}\ln \mathcal W_l(t) = \Gamma_1'(l)P_1(t)\,\mathrm{d}t + \Gamma_1(l)\mathrm{d}P_1(t).
\end{equation}
As discussed in the same subsection and shown on Figure~\ref{fig:PCA}, the level factor is close to constant: $\gamma_{1l} = 1$ (without loss of generality) and $\Gamma_1(l) = l$. Plugging this into~\eqref{eq:ct-ret-pc1}:
\begin{equation}
\label{eq:ct-ret-l}
\mathrm{d}\ln \mathcal W_l(t) = P_1(t)\,\mathrm{d}t + l\mathrm{d}P_1(t).
\end{equation}
For example, letting $l = 0$, we get:
\begin{equation}
\label{eq:ct-ret-0}
\mathrm{d}\ln \mathcal W_0(t) = P_1(t)\,\mathrm{d}t.
\end{equation}
Here, $\mathcal W_0(t)$ as wealth invested in very short-term Treasury bonds, which are risk-free. They do not have any interest rate risk. These are usually taken to be the benchmark for risk-free returns. Subtracting~\eqref{eq:ct-ret-0} from~\eqref{eq:ct-ret-l}, we get the formula for {\it term premium:} Extra returns of long-term Treasuries which are compensation for interest rate risk:
\begin{equation}
\label{eq:term-premia}
\mathrm{d}\ln\frac{\mathcal W_l(t)}{\mathcal W_0(t)} = \mathrm{d}\ln\mathcal W_l(t) - \mathrm{d}\ln\mathcal W_0(t) = l\cdot \mathrm{d}P_1(t).
\end{equation}
The benchmark can be taken as bonds of fixed maturity $l_0$, for example 10 years (classic long-term Treasury benchmark). Then from~\eqref{eq:term-premia}, we have:
\begin{equation}
\label{eq:CAPM}
\mathrm{d}\ln\frac{\mathcal W_l(t)}{\mathcal W_0(t)} = \frac{l}{l_0}\cdot\mathrm{d}\ln\frac{\mathcal W_{l_0}(t)}{\mathcal W_0(t)}.
\end{equation}
We can interpret~\eqref{eq:CAPM} as follows: Only one factor, namely the maturity $l$, fully explains the term premium. This is similar to the classic Capital Asset Pricing Model (CAPM), where the equity premium of a stock portfolio (its total returns minus risk-free returns) equals the equity premium of a benchmark (usually S\&P 500) times $\beta$. For background on CAPM, see \cite[Chapter 6, Section 3]{Ang}. According to CAPM, only one factor matters for stock portfolio returns: {\it market exposure} $\beta$. Other risk might be diversified away. 

CAPM remains a useful benchmark for portfolio management. However, empirical evidence works again CAPM, \cite[Chapter 6, Section 5]{Ang}. There are other important factors. Most commonly cited are {\it size} and {\it value}, see \cite[Chapter 7, Section 3]{Ang}. In the Treasury bond market, we can add more factors by including more PC, such as slope and curvature. Adding the level and the slope, assume the loading factor function is linear: $\gamma_{2l} = l + c$ and $\Gamma_2(l) = l^2 + cl$. Then~\eqref{eq:ct-returns} becomes
$$
\mathrm{d}\ln\mathcal W_l(t) = P_1(t)\,\mathrm{d}t + l\,\mathrm{d}P_1(t) + (2l+c)\,P_2(t)\,\mathrm{d}t + (l^2 + cl)\,\mathrm{d}P_2(l). 
$$
Letting $l = 0$ and subtracting, we can compute the term premium:
$$
\mathrm{d}\ln\frac{\mathcal W_l(t)}{\mathcal W_0(t)} = l\,\mathrm{d}P_1(t) + 2l\,P_2(t)\,\mathrm{d}t + (l^2+cl)\,\mathrm{d}P_2(t).
$$
Here, term premium is not proportional to $l$: Note the term $l^2$ in the right-hand side. 

\section{Conclusion and Discussion}

We introduced in this article a novel multivariate autoregressive stochastic volatility model, which is a generalization of existing models. This model is an alternative to ARMA-GARCH models, with two series of innovations instead of one. We fit this model for a few first PC (level, slope, and curvature) for zero-coupon Treasury bonds, with volatility directly observed as VIX. We also prove long-term stability for our model and for its continuous-time version, and the Strong Law of Large Numbers. Finally, for these zero-coupon bonds, we can easily model total returns, and prove long-term stability for them as well. Our main observations:

\begin{itemize}
\item The volatility term is directly observed here, unlike GARCH or other similar models, where volatility is hidden and must be inferred from the main data series.
\item Stock market implied volatility is useful for the bond market. This is a remarkable connection between the American stock and bond markets.
\item Including VIX is useful for slope (the second PC, showing how long-term rates differ from short-term rates) but not for other PC (most importantly, for the level, which shows how rates move together). 
\end{itemize}

Further research might include extension to ARMA model, used for the main data series, and stochastic volatility is used in modeling noise. We might also use a more complicated model, such as ARMA, rather than autoregression of order 1, for log volatility. 

Innovations might be Gaussian, Laplace, variance-gamma (a generalization of the Laplace distribution with MGF $e^{ht}(1 - ct - at^2)^{-b}$), or having Student $t$-distribution. 

We can assume volatility is unobservable. Then we would need to infer it using hidden Markov models. This might be done in frequentist of Bayesian setting. This research will probably be quite computationally intensive. 

One can state and prove exponential convergence to the stationary distribution in the long run for our model or its generalizations. It would be great to find or at least estimate this exponential convergence rate, similarly to \cite[Chapter 19]{Stability}. 

We might consider daily or weekly instead of monthly data. There is much more data, but it is likely harder to model. Conversely, switching to annual data leaves us less data, but likely easier to model, using simpler models. Innovations will be likely closer to IID Gaussian for annual rather than monthly data.

Future research might use other bond market data: Classic coupon Treasury bonds; corporate bonds with various ratings; and international bonds. As discussed earlier, it is much harder to model coupon-paying Treasury bond returns. But one could consider exchange-traded funds (ETF) tracking Treasury bonds.

A proof of the Strong Law of Large Numbers: convergence almost surely (instead of convergence in probability) in Theorem~\ref{thm:returns-continuous} seems to require more effort, which we leave for future research. Note that in other statements of this article, the Law of Large Numbers is shown in the strong version. 

In subsection 2.11, eigenvalues of $\mathbf{B}$ for the bivariate and the trivariate models are between 0 and 1, but very close to 1. It might be interesting to apply unit root tests to these multivariate models, and check whether we can statistically verify that all eigenvalues of $\mathbf{B}$ are inside the unit circle. If we fail to reject the unit root hypothesis: some eigenvalues of $\mathbf{B}$ lie on the inut circle, then we do not have stationarity and ergodicity. 

\section*{Appendix} 

Here we state and prove a few technical lemmas which seem quite clear to us, but we could not find rigorous complete proofs in existing literature. 

\begin{lemma} Consider a discrete-time Markov process $\mathbf{X} = (\mathbf{X}(t), t = 0, 1, 2, \ldots)$ on $\mathbb R^d$ with transition density $p : \mathbb R^{2d} \to (0, \infty)$: For any $B \subseteq \mathbb R^d$ and $t > 0$, $\mathbf{x} \in \mathbb R^d$, 
$$
\mathbb P(\mathbf{X}(t) \in B\mid \mathbf{X}(t-1) = \mathbf{x}) = \int_Bp(\mathbf{x}, \mathbf{y})\,\mathrm{d}\mathbf{y}.
$$
If this process has a stationary distribution $\pi$, it is ergodic. 
\label{lemma:ergodic-discrete}
\end{lemma}

\begin{proof} We rely on the monograph \cite{Stability} and the survey \cite{Rosenthal}. 

{\it Step 1.} Let us show that this Markov process is {\it irreducible}, in terms of \cite[page 31]{Rosenthal} or \cite[subsection 3.4.3, section 4.2]{Stability}. With positive probability, any set of positive Lebesgue measure can be hit starting from any initial point, even at the first step, since the integral of this positive density with respect to this set is strictly positive. 

{\it Step 2.} Next, let us show $\pi$ is {\it equivalent} (mutually absolutely continuous) with respect to the Lebesgue measure on $\mathbb R^d$: In other words, let us show these two measures have the same null sets. Indeed, another way to write that $\mathbf{X}(0) \sim \pi$ implies $\mathbf{X}(1) \sim \pi$ is:
$$
\pi(A) = \int_{\mathbb R^d}\int_Ap(\mathbf{x}, \mathbf{y})\,\mathrm{d}\mathbf{y}\,\mathrm{d}\pi(\mathbf{x}).
$$
Equivalence follows from here immediately, using strict positivity of $p$. 

{\it Step 3.} Further, let us show this Markov process is {\it aperiodic} in the sense of \cite[Theorem 5.4.4]{Stability} or \cite[page 32]{Rosenthal}. Assume the converse: $\mathbb R^d$ is split into $s \ge 2$ subsets $D_0, \ldots, D_{s-1}$, and 
$$
\mathbf{X}(t) \in D_j \Rightarrow \mathbf{X}(t+1) \in D_{(j+1)\,\mathrm{mod}\, s}.
$$
From positivity of the transition density $p$, each set $\mathbb R^d\setminus D_j$ has Lebesgue measure zero. This contradicts $D_0\cup \ldots\cup D_{s-1} = \mathbb R^d$. 

{\it Step 4.} Next, let us show this Markov process is {\it Harris recurrent} with respect to the Lebesgue measure on $\mathbb R^d$, see \cite[Proposition 9.1.1]{Stability}. Since we proved irreducibility earlier, we need only to show that, for any set $A \subseteq \mathbb R^d$, starting from any $\mathbf{x} \in \mathbb R^d$, the process will almost surely hit $A$. This follows from positivity of transition density: The process will hit $A$ with positive probability. It remains to apply \cite[Lemma 20]{Rosenthal} to conclude that the process will hit $A$ almost surely, starting from almost every $\mathbf{x} \in \mathbb R^d$.

Finally, let us remove the words {\it almost every} above. Denote by $\mathcal N \subseteq \mathbb R^d$ the set of such exceptional starting points. If we start from any $\mathbf{y} \in \mathcal N$, at the first step we will almost surely get to the set $\mathbb R^d\setminus \mathcal N$. Starting from there, we will hit $A$ almost surely. Therefore, we proved for any initial condition we hit any given set of positive Lebesgue measure. 

Thus we showed that this discrete-time Markov process is irreducible, aperiodic, and Harris recurrent. Apply \cite[Theorem 13.0.1]{Stability} and complete the proof of ergodicity. 
\end{proof}

\begin{lemma}
\label{lemma:ergodic-continuous}
Take a $d$-dimensional diffusion $\mathbf{X} = (\mathbf{X}(t),\, t \ge 0)$: a solution of a stochastic differential equation with $C^{\infty}$ drift vector and $C^{\infty}$ nonsingular diffusion matrix. If this process is tight, then it is ergodic. 
\end{lemma}

\begin{proof} In \cite[subsection 2.3]{MT1993a}, the discrete-time process $(\mathbf{X}(ns), n = 0, 1, 2, \ldots)$, for any time step $s > 0$, is called a {\it skeleton Markov chain}. In terms of \cite[subsection 3.1]{MT1993a}, tightness is called {\it boundedness in probability} (topological stability condition 2), and is stronger than {\it non-evanescence} (topological stability condition 1): $\mathbf{X}(t) \to \infty$ as $t \to \infty$ with probability zero, for any initial condition. However, in terms of \cite[Section 3]{Bhattacharya}, this implies each point in $\mathbb R^d$ is {\it recurrent}. Applying to \cite[Theorem 3.1(d)]{Bhattacharya}, any open set $U \subseteq \mathbb R^d$ is eventually hit almost surely by any skeleton Markov chain. 

Note that the diffusion matrix has nonzero determinant at any point in $\mathbb R^d$. Therefore, its column vectors span the entire $\mathbb R^d$. We can consider each column vector as a function $\mathbb R^d \to \mathbb R^d$, or, further, as a $C^{\infty}$ vector field on $\mathbb R^d$, see \cite[Chapter 1, Section 3]{Kolar}. Recall that a Lie algebra is also a vector space; see any Lie algebra textbook, for example \cite[Section 3.1]{Hall}. Therefore, the Lie algebra generated by these column vector fields is also equal to $\mathbb R^d$. By \cite[Theorem 3.3(ii)]{MT1993a}, this diffusion is a {\it $T$-process} in the sense of the definition in \cite[subsection 3.2]{MT1993a}. What is more, any skeleton Markov chain is also a {\it $T$-process} in the sense of \cite[Chapter 6]{Stability}. More precisely, it has a {\it strong Feller property}: for any bounded measurable function $f : \mathbb R^d \to \mathbb R$, and any $s > 0$, the function $P^tf$ defined by
$$
P^sf(\mathbf{x}) = \mathbb E(f(\mathbf{X}(s)\mid \mathbf{X}(0) = x)
$$
is bounded and continuous. A weaker version of this property: For bounded and continuous $f$, $P^sf$ is also bounded continuous  is called {\it weak Feller property}. 

Apply \cite[Theorem 6.0.1]{Stability} and conclude that there exists a $\sigma$-finite measure $\psi$ on $\mathbb R^d$ with the following property: Any skeleton chain hits almost surely not only any open set, but any set $A$ with $\psi(A) > 0$. In practice, 
$\psi$ is Lebesgue measure, but this is not important here. Harris recurrence of the original continuous-time process (called {\it stochastic stability condition 1}) in the sense of \cite[subsection 2.2]{MT1993a} follows from here. Apply \cite[Theorem 3.1]{MT1993a}: this process $\mathbf{X}$ has a stationary distribution. Therefore, $\mathbf{X}$ is {\it positive Harris recurrent} (called {\it stochastic stability condition 2}). Finally, apply \cite[Theorem 6.1]{MT1993a} and complete the proof. 
\end{proof}

\begin{lemma}
Take a $d$-dimensional diffusion $\mathbf{X} = (\mathbf{X}(t),\, t \ge 0)$: a solution of a stochastic differential equation with continuous drift vector and continuous nonsingular diffusion matrix. If there exists a function $U : \mathbb R^d \to [0, \infty)$ with $U(\infty) = \infty$ such that, regardless of $\mathbf{X}(0)$, 
\begin{equation}
\label{eq:BDD}
\mathcal C := \sup\limits_{t \ge 0}\mathbb E[U(\mathbf{X}(t))] < \infty,
\end{equation}
then this diffusion is tight, ergodic, with unique stationary distribution $\pi$. For any measurable function $f : \mathbb R^d \to \mathbb R$ such that $|f| \le U$, almost surely, regardless of the initial condition, 
$$
\frac1T\int_0^Tf(\mathbf{X}(t))\,\mathrm{d}t \to \int_{\mathbb R^d}f\mathrm{d}\pi,\quad T \to \infty.
$$
\label{lemma:SLLN}
\end{lemma}

\begin{proof} {\it Step 1.} We use \cite[Section 8]{MT1993a}. It does statements for general reducible Markov processes. Here, the diffusion process is irreducible. Therefore, there is only one $\tilde{H}_i$: $|I| = 1$. Thus, the {\it random probability} is equal to the stationary distribution: $\tilde{\pi} = \pi$; and the {\it invariant transition function} is also equal to $\Pi = \pi$. The $L^1$ space in \cite[Theorem 8.1(i)]{MT1993a} is thus equal to $L^1(\mathbb R^d, \pi)$. 

{\it Step 2.} The process is tight (bounded in probability on average). This follows from~\eqref{eq:BDD}. Indeed, pick a $\varepsilon > 0$. Since $U(\infty) = \infty$, the set $\mathcal K := \{\mathbf{x} \in \mathbb R^d\mid U(\mathbf{x}) \le C/\varepsilon\}$ is compact. Use the Markov inequality \cite[Proposition 5.1.1]{Look}: For any $t \ge 0$, $\mathbb P(\mathbf{X} \in \mathcal K) \ge 1 - \varepsilon$. 

{\it Step 3.} We need also to show that 
\begin{equation}
\label{eq:integrable}
\int_{\mathbb R^d}U(\mathbf{x})\pi(\mathrm{d}\mathbf{x}) \le \mathcal C.
\end{equation}
Apply Lemma~\ref{lemma:ergodic-continuous} to conclude that the process is ergodic. Therefore, $\mathbf{X}(n) \to \pi$ as $n \to \infty$: The skeleton chain converges to the stationary measure in total variation, and therefore weakly. Apply \cite[Proposition 10.2.2]{Look} and use~\eqref{eq:BDD} to conclude~\eqref{eq:integrable}. 

{\it Step 4.} From the proof of Lemma~\ref{lemma:ergodic-continuous}, this diffusion is a $T$-process. We showed in Step 2 that this process is bounded in probability on average. For any $f : \mathbb R^d \to \mathbb R$ such that $|f| \le U$, we have: $f \in L^1(\mathbb R^d, \pi)$. Use results of Step 1 and \cite[Theorem 8.1]{MT1993a}.
\end{proof}

\begin{lemma}
Take a vector-valued square-integrable martingale $\mathcal M$ in $\mathbb R^d$ such that its $i$th component has $\mathrm{d}\langle \mathcal M_i\rangle_t = \mathcal G_{i}(t)\,\mathrm{d}t$, and a $d\times d$-matrix-valued bounded deterministic function $\mathcal F : \mathbb R \to \mathbb R^{d\times d}$. Define the vector-valued stochastic integral componentwise:
\begin{align}
\label{eq:i-th}
\begin{split}
\mathcal F\cdot\mathcal M(t) &:= \int_0^t\mathcal F(s)\,\mathrm{d}\mathcal M(s),\quad t \ge 0,\\
(\mathcal F\cdot\mathcal M)_i(t) &:= \sum_{j=1}^d\int_0^t\mathcal F_{ij}(s)\mathrm{d}\mathcal M_j(s),\quad i = 1, \ldots, d,\quad t \ge 0.
\end{split}
\end{align}
Then $\mathcal F\cdot\mathcal M = (\mathcal F\cdot\mathcal M(t),\, t \ge 0)$ is also a square-integrable martingale, and 
$$
\mathbb E|F\cdot\mathcal N(t)|_2^2 \le d^{3/2}\int_0^t|\!|F(s)|\!|_S^2\cdot\sum \mathcal G_i(s)\,\mathrm{d}s.
$$
\label{lemma:matrix-mgle}
\end{lemma}

\begin{proof} The quadratic variation of the $i$th component of this stochastic integral from~\eqref{eq:i-th} is
\begin{equation}
\label{eq:q-v}
\langle (\mathcal F\cdot\mathcal M)_i\rangle_t = \int_0^t\sum\limits_{j=1}^d\sum\limits_{j'=1}^d\mathcal F_{ij}(s)\mathcal F_{ij'}(s)\,\mathrm{d}\langle \mathcal M_i, \mathcal M_j\rangle_s.
\end{equation}
The matrix is symmetric and nonnegative definite for any $0 \le s \le t$:
$$
\mathcal A(s, t) := \left(\langle \mathcal M_j, \mathcal M_{j'}\rangle_t - \langle \mathcal M_j, \mathcal M_{j'}\rangle_s\right)_{jj'}.
$$  
This follows from the observation that for any constant vector $\mathbf{v} \in \mathbb R^d$, we have: 
$$
\mathcal A(s, t)\mathbf{v}\cdot\mathbf{v} = \langle \mathcal M\cdot \mathbf{v}\rangle_t - \langle \mathcal M\cdot \mathbf{v}\rangle_s \ge 0.
$$
Using~\eqref{eq:positive-F}, we conclude that the Frobenius norm of this matrix is dominated by its trace times a constant $\sqrt{d}$: 
$$
d^{1/2}\sum\limits_{j=1}^d\left(\langle \mathcal M_j\rangle_t - \langle \mathcal M_j\rangle_t\right) = d^{1/2}\sum\limits_{j=1}^d\int_s^t\mathcal G_j(u)\,\mathrm{d}u.
$$
Approximate the integral~\eqref{eq:q-v} using Riemann-Stieltjes partial sums. Take expectation and sum over all $i = 1, \ldots, d$. Apply the norm $|\cdot|_2$ and the triangle inequality. Get back from the sums to the integral by letting the mesh tend to zero. 
\begin{equation}
\label{eq:upper-bound}
\mathbb E|(\mathcal F\cdot\mathcal N)(t)|^2_2 \le \mathbb E\int_0^t\mathrm{tr}(\mathcal F(s)\mathcal G(s)\mathcal F^T(s))\,\mathrm{d}s.
\end{equation}
Use the Cauchy inequality twice in~\eqref{eq:Cauchy} to get:
\begin{equation}
\label{eq:C2}
|\mathrm{tr}(\mathcal F(s)\mathcal G(s)\mathcal F^T(s))| \le |\!|\mathcal F(s)|\!|_F^2\cdot |\!|\mathcal G(s)|\!|_F. 
\end{equation}
From~\eqref{eq:norms},~\eqref{eq:upper-bound},~\eqref{eq:C2}, we complete the proof. 
\end{proof}

\end{document}